\documentclass{amsart}
\usepackage{graphicx, color, enumerate}
\usepackage[square,numbers]{natbib}
\usepackage{natbib, url}

\newcommand\q{\mathbf{q}}
\newcommand\s{\mathbf{s}}
\newcommand\Q{\mathcal{Q}}

\newcommand\Shape{\mathcal{S}}
\newsavebox{\astrutbox}
\sbox{\astrutbox}{\rule[-5pt]{0pt}{20pt}}

\newtheorem{theorem}{Theorem}
\newtheorem{prop}{Proposition}

\title{Generalized Scallop Theorem for Linear Swimmers}

\author{T. Chambrion}
\address{Institut \'Elie Cartan UMR 7502, Nancy-Universit\'e,
CNRS, INRIA, B.P.~239, F-54506 Vandoeuvre-l\`es-Nancy Cedex,
France, and INRIA
Lorraine, Projet CORIDA}
\email{thomas.chambrion@iecn.u-nancy.fr}

\author{A. Munnier}
\address{Institut \'Elie Cartan UMR 7502, Nancy-Universit\'e,
CNRS, INRIA, B.P.~239, F-54506 Vandoeuvre-l\`es-Nancy Cedex,
France, and INRIA
Lorraine, Projet CORIDA}
\email{alexandre.munnier@iecn.u-nancy.fr}
\thanks{Authors both supported by CPER MISN AOC. First author supproted by ANR GCM, ERC Boscain and BQR Lorraine, and second author by ANR CISIFS and ANR GAOS}

\begin{document}

\begin{abstract}
In this article, we are interested in studying locomotion strategies for a class of shape-changing bodies swimming in a fluid. This class consists of swimmers subject to a particular linear dynamics, which includes the two most investigated limit models in the literature: swimmers at low and high Reynolds numbers. Our first contribution is to prove that although for these two models the locomotion is based on very different physical principles, their dynamics are similar under symmetry assumptions. Our second contribution is to derive for such swimmers a purely geometric criterion allowing to determine wether a given sequence of shape-changes can result in locomotion. This criterion can be seen as a generalization of Purcell's scallop theorem (stated in \citep{Purcell:1977aa}) in the sense that it deals with a larger class of swimmers and address the complete locomotion strategy, extending the usual formulation in which only periodic strokes for low Reynolds swimmers are considered. 
\end{abstract}
\maketitle
\section{Introduction}
\subsection{About Purcell's theorem}
The specificity of swimmers at low Reynolds numbers (like microorganisms) is that inertia for both the fluid and the body can be neglected in the equations of motion. Consequently, as highlighted by Purcell in his seminal article \citep{Purcell:1977aa}, the mechanisms they used to swim are quite counter-intuitive and can give rise to surprising phenomena, the most famous one being undoubtedly illustrated by the so-called scallop theorem. Roughly speaking, this theorem states that periodic strokes consisting of {\it reciprocal} shape-changes (i.e. a sequence of shape-changes invariant under time reversal) cannot result in locomotion (i.e. does not allow to achieve a net displacement of arbitrary length) in a viscous fluid.  Considering the prototypical example of the scallop (as sketched on the left of Fig.~\ref{purcell}), which can only open and close its shell, and assuming that the animal lives in a low Reynolds environment (which it does not), Purcell explains that {\it it can't swim because it only has one hinge, and if you have only one degree of freedom in configuration space, you are bound to make a
reciprocal motion. There is nothing else you can do}. In addition to the light this result casts on the understanding of the hydrodynamics of swimming microorganisms, it has to be taken into account as a serious pitfall for the design of micro-robots, for which engineers' interest grows along with the number of applications that have been envisioned for them (such as, for instance, drug deliverers in the area of biomedicine).
\begin{figure}
\centerline{\input{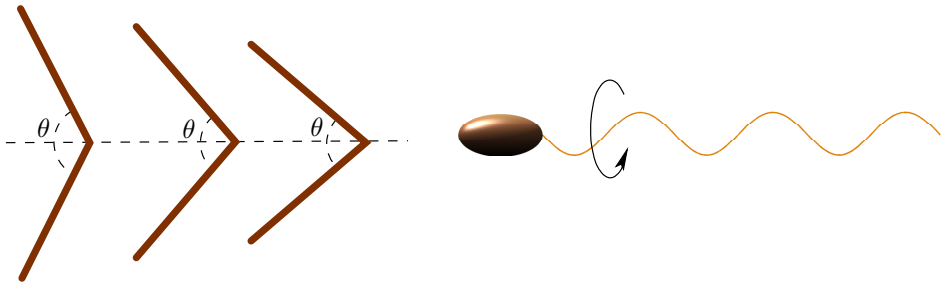}}
\caption{\label{purcell}On the left, Purcell's 2D scallop cannot swim in a viscous fluid... and neither can it in a potential flow. On the right, 3D Purcell's corkscrew can swim in a viscous fluid and probably also in a potential one.}
\end{figure}

The classical assumptions of Purcell's theorem are that the shape-changes have to be time periodic and the {\it sequence of shapes} (over a stroke), invariant under time reversal. Notice that the latter condition does not mean that the shape-changes have to be strictly time-reversal invariant, with the same forward and backward rate, but only that the succession of shapes is the same when viewed forward and backward in time. Under these hypotheses, Purcell concludes that the swimmer comes back to its initial position after performing a stroke. Going through Purcell's article, one will find no proof for this result. However, a huge literature devoted to this topic has been produced since then and mathematical proofs can be found, for instance, in the article of E. Lauga and T.R. Powers \citep{Lauga:2009aa} (which contains also an impressive list of references and to which we refer for a comprehensive bibliography on this topic) and in  \citep{DeSimone:2009aa} by DeSimone et al. 

\subsection{Beyond Purcell's theorem}
Although, as already mentioned, Purcell does not provide a rigorous proof of his famous theorem, he explains that the keystone of his result relies on that inertia is not taking into account in the modeling of low Reynolds swimmers, allowing in particular the Navier-Stokes equations governing the fluid flow to be simplified into the steady Stokes equations. Our first main contribution in this article will be to prove that more widely, Purcell's theorem in its original form still holds true for a class of swimmers subject to a particular linear dynamics that will be made precise later on. This class obviously includes low Reynolds swimmers but also high Reynolds swimmers extensively studied in the literature (see for instance the article \citep{Kanso:2005aa} of E. Kanso et al. or \citep{Chambrion:2010aa} by T. Chambrion and A. Munnier, and references therein).

Purcell's theorem does not admit any reverse statement allowing to determine wether a sequence of shape-changes violating the hypotheses can result in locomotion. To illustrate this idea, consider Fig.~\ref{non_periodic} on which is plotted the graphs of the functions $t\in\mathbf R_+\mapsto\theta_j(t)$ (for $j=1,2$), where $t$ stands for the time and each $\theta_j(t)$ ($j=1,2$) gives the value of the angle of the scallop's hinge, as sketched on the left of Fig.~\ref{purcell}. 
\begin{figure}
\centerline{\includegraphics[width=.99\textwidth]{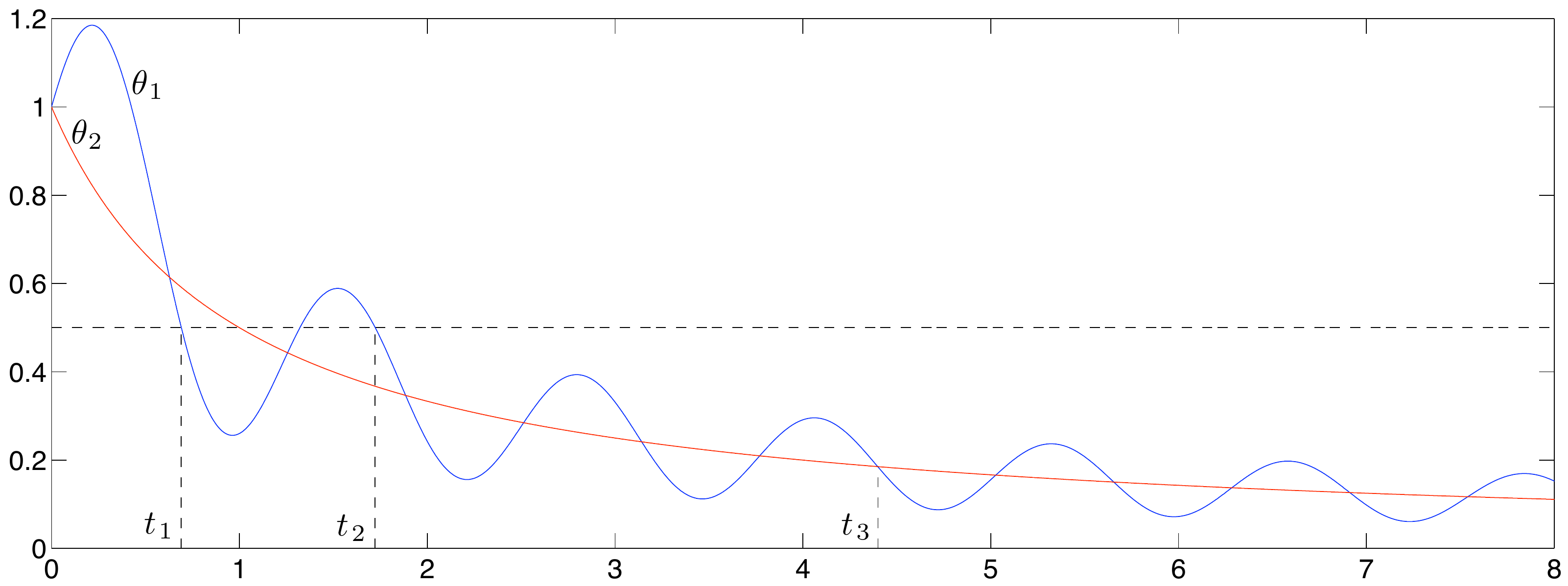}}
\caption{\label{non_periodic}Assume that the curves above give the value of the scallop's hinge angle $\theta$ (see the left hand side of Fig.~\ref{purcell}) with respect to time. Although these shape-changes are neither time reversal invariant nor time periodic, the scallop is at the same position at the times $t_1$ and $t_2$. Besides, the displacement of the scallop between the times $t=0$ and $t=t_3$ depends only on the values of $\theta(0)$ and $\theta(t_3)$ and not on the shape of the curve in between. }
\end{figure}
None of these sequences of shape-changes is neither periodic nor time reversal invariant. However, anybody familiar enough with Purcell's result would agree that the scallop undergoing the shape-changes corresponding to the function $\theta_1$ will not move on average, between the times $t_1$ and $t_2$. Likewise, the mollusk will be at the same place at the time $t_3$ after performing either sequence corresponding to $\theta_1$ or $\theta_2$. One may also wonder where the animal would go asymptotically, as time goes to infinity. Following Purcell's reasoning, probably not very far and more precisely, exactly at the same distance as if the angle  would have ranged from 1 to 0 over a finite time interval... because time does not matter in the low Reynolds world. This last property suggests that the hypotheses of the theorem could be restated in a purely geometric framework and one may even think at this point that, sticking to the scallop example, a reasonable statement could be something like: {\em the displacement of the scallop is a continuous function of the angle range}. As an obvious consequence, one would deduce that a bounded angle range implies a bounded displacement. This is true but unfortunately cannot be extended to the general case. Indeed, consider now an other example of swimmer, pictured on the right of Fig.~\ref{purcell} and called by Purcell {\it the corkscrew} (and whose way of swimming is quite obvious). The configuration space is the one dimensional torus $\mathcal S^1$ and the rotation of the flagella is known to produce a net displacement of the hypothetic animal. On Fig.~\ref{rotating_flagella} is drawn the graph of a function giving the value of the angle of rotation, valued in $\mathcal S^1$, with respect to the time. 
\begin{figure}
\centerline{\includegraphics[width=.8\textwidth]{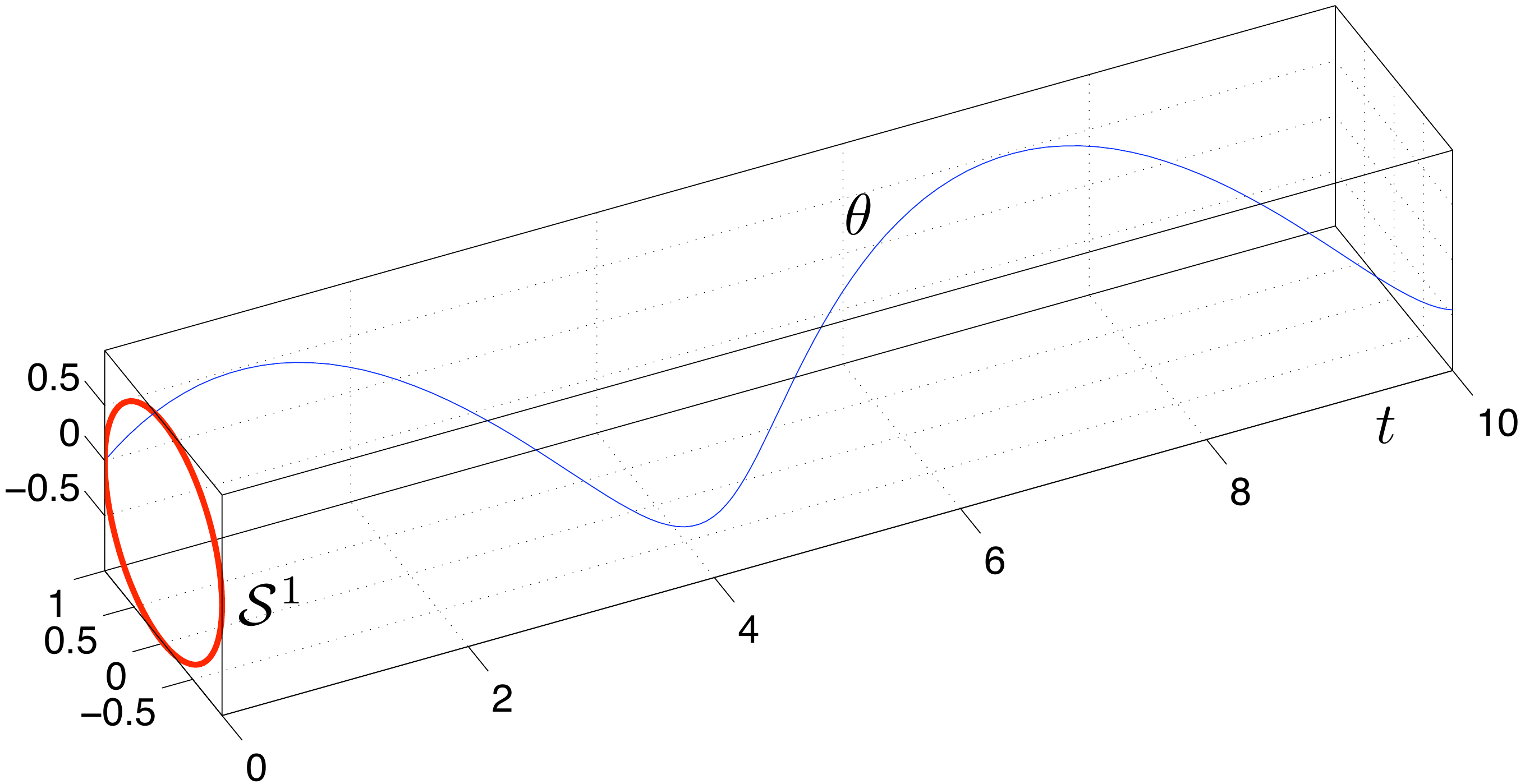}}
\caption{\label{rotating_flagella}The graph gives the angle of rotation $\theta$ of the corkscrew's flagella (pictured on the right hand side of Fig.~\ref{purcell}) valued in the one dimensional torus $\mathcal S^1$ with respect to time. By looking only at the projection of the curve on $\mathcal S^1$, one cannot determine how many tours have been performed. This is an evidence that although Purcell's scallop theorem can be addressed using a purely geometric point of view, it cannot be done without the recourse to the notion of universal cover.}
\end{figure}
It is not so easy to reiterate the exercise of the preceding example and to derive a purely geometric criterion (i.e. time independent) allowing to determine wether the displacement is bounded or not. The reason is that, by looking only at $\mathcal S^1$, it is not possible to determine how many tours have been performed by the flagella. To do so, we have to look at the angle as valued not in $\mathcal S^1$ but in the universal cover of the manifold. The notion of {\it universal cover} will allow us to state a generalized and purely geometric version of the scallop theorem, which will be the second main contribution of the paper.

\subsection{Outline of the paper}
In Section 2, we present
an abstract framework and state
a generalized scallop theorem for a class of shape-changing bodies, called {\it linear swimmers}.
This is quite classical material, except for a topological interpretation
 of what a {\it reciprocal motion} is, which may be original.
In Section 3, we prove that swimmers
at low Reynolds numbers and high Reynolds numbers (in a potential fluid and with some symmetry assumptions)
are linear swimmers and meet the requirements of our main theorem. Finally, a numerical simulation of a swimmer
in a perfect fluid is given in Section 4.

\section{Abstract result}
\subsection{General assumptions on the swimmer}
We assume that any possible shape of the swimmer can be described by a so-called {\it shape variable} $\s$ 
living in a Banach space $\mathcal S$ (which can be infinite dimensional).
So the shape-changes are described by means of a smooth {\it shape function} $t\in\mathbf R_+\mapsto\mathbf s(t)\in\mathcal S$ where $t$ stands for the time and $\dot{\mathbf s}=d\mathbf s/dt$ is the rate of change. The variable $\q\in\mathcal Q$, where $\Q$ is a
smooth,  finite dimensional Riemannian manifold, gives the position of the swimmer in the fluid.
For instance, to describe the position of Purcell's scallop, we would choose $\mathcal Q=\mathbf R$ because the scallop can only move along a straight line. Since its shape is thoroughly described by the angle $\theta$, we would have $\mathbf s=\theta$ and $\mathcal S=\mathbf{R}$ or $\Shape=\mathbf R/2\pi$ as well. 
\begin{figure}
\centerline{\input{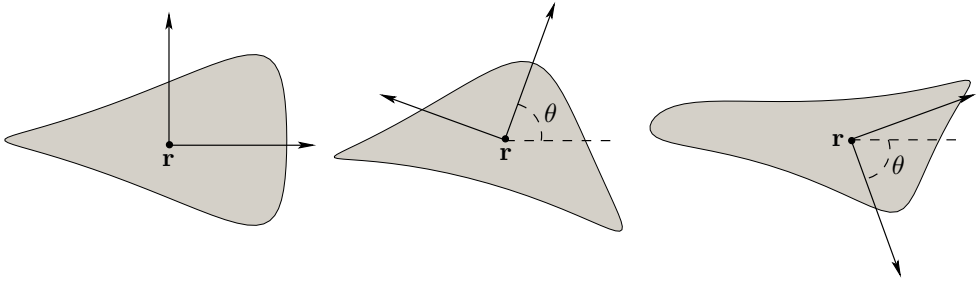}}
\caption{\label{TandA}Three examples of shapes of the authors' 2D-{\it amoeba} model described in \citep{Chambrion:2010aa}. The frame is attached to the body and moves along with it. The shape variable $\s$ is here a complex sequence and $\mathbf q=(\mathbf r,R(\theta))\in\mathbf R^2\times{\rm SO}(2)$ where $\mathbf r$ is the position of the center of mass of the swimmer (expressed in a fixed Galilean frame) and $R(\theta)$ is a rotation matrix of angle $\theta$, giving its orientation.}
\end{figure}
In \citep{Chambrion:2010aa},
we give an example of 2D-swimmer (see Fig.~\ref{TandA}) in an infinite
extent of perfect fluid with potential flow.
In this case, the
manifold $\Q$ is $\mathbf{R}^2 \times {\rm SO}(2)$, while the shape space $\Shape$
is an infinite dimensional Banach space, consisting of complex sequences $\mathbf s=({s}_k)_{k\geq 1}$ ($s_k\in\mathbf C$, $k\geq 1$) and endowed with the norm 
$\|\mathbf s\|_{\mathcal S}=\sum_{k=1}^\infty k |s_k|$.

Notice however that  physical and mathematical constraints usually affect the pair $(\mathbf s,\dot{\mathbf s})$ and lead to the definition of {\it allowable shape function}. It entails in particular  that $\mathbf s$ is bound to remain in a subset of $\mathcal S$ and that $\dot{\mathbf s}$ cannot take any value in $\mathcal S$ either. As an example, let us mention the constraint of self-propulsion, which means that although directly prescribed, the shape-changes have to result from the work of hypothetical internal forces, occurring within the swimmer (like for instance the work of muscles). This constraint prevents, for instance, translations 
to be considered as possible shape-changes. At this point, we also add the constraint that the path $\gamma=\mathbf s(\mathbf R_+)$ be included in $X$, a one dimensional submanifold immersed in $\mathcal S$. From a physical point of view, it means that at any moment, there is only one degree of freedom in the shape-changes (see Fig.~\ref{non_crossing}).
\begin{figure}
\centerline{\input{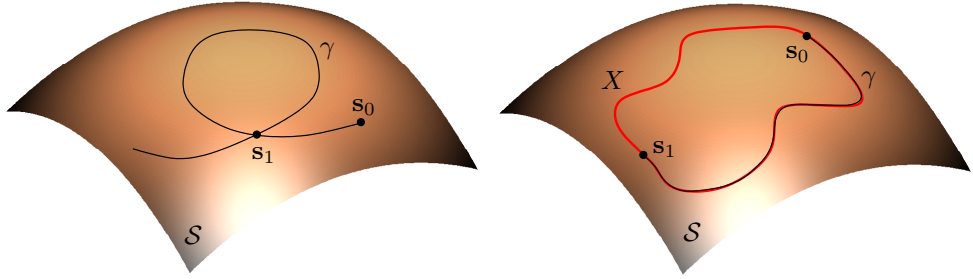}}
\caption{\label{non_crossing}On the left, the path $\gamma$ is not included in a one dimensional submanifold of $\mathcal S$. It means that for both moments corresponding to $\mathbf s=\mathbf s_1$, there are two degrees of freedom for the shape-changes (each one corresponding to a branch). On the right, notice that the submanifold $X$ can be bounded or not.}
\end{figure}
We recall that two complete examples of modeling (the low and the high Reynolds swimmers) are given in Section~\ref{SEC:modelling}.
\subsection{Dynamics of linear swimmers}\label{SEC_assumption_linear_dynamics}
Denote by $T_\q\Q$ the tangent space to $\Q$ at the point $\q$, by $T\Q$ the tangent bundle to $\Q$ and by ${\mathcal L}(\Shape,{T\Q})$  the space 
of the linear mappings from $\Shape$ to $T\Q$. We call {\it linear swimmer}, any model of shape-changing body whose dynamics has the form:
\begin{equation}\label{EQ_main}
 \frac{d}{dt}
\mathbf{q}(t)=\langle\mathbf{F}(\mathbf{q},\mathbf{s}),\dot{\mathbf{s}}
(t)\rangle,\qquad(t>0),
\end{equation}
where $\mathbf{F}:\Q \times \Shape \rightarrow {\mathcal L}(\Shape,{T\Q})$ is a smooth function satisfying 
\begin{enumerate}[(i)]
\item $\langle\mathbf{F}(\mathbf{q},\s),\tilde{\mathbf s}\rangle\in T_{\mathbf{q}} \Q$ for every
$\q\in\Q$ and every $\s,\tilde\s\in\Shape$;
\item There exists $K>0$ such that $\|\langle\mathbf{F}(\mathbf{q},\mathbf{s}),\tilde{\mathbf{s}} \rangle \|_{{T_\q\Q}}\leq K \|\s\|_{\mathcal S} \| \tilde{\s}\|_{\mathcal S}$ for every
$\q\in\Q$ and every $\s,\tilde\s\in\Shape$.
\end{enumerate}
The Cauchy-Lipschitz theorem guarantees that,
for any $\q_0\in\Q$ and for any smooth allowable
shape function $\mathbf{s}:\mathbf R_+\to \Shape$ there exists a unique solution to
(\ref{EQ_main}) with Cauchy data $\q(0)=\q_0$. 

\subsection{Generalized scallop theorem}

The main feature of linear swimmers' dynamics is the following reparameterization property.
\begin{prop}\label{THE_reparam}
Let any allowable control function $\mathbf s:\mathbf{R_+}\rightarrow
\Shape$  and any point $\mathbf{q}_0 \in \Q$ be given. Denote by
$\mathbf{q}:\mathbf{R_+}\rightarrow \Q$ the solution to Equation (\ref{EQ_main}) with initial
condition  $\mathbf{q}_0$. Then, 
for any  $\mathcal C^1$ function $\beta:\mathbf{R}_+
\rightarrow \mathbf{R_+}$,  the solution
$\mathbf{q}_{\beta}:\mathbf{R}
_+\rightarrow  \Q$ to Equation (\ref{EQ_main}) corresponding to the
shape-changes $\mathbf{s}_{\beta}:=t\in\mathbf R_+\mapsto \mathbf s(\beta(t))\in\Shape$,
with
initial condition $\mathbf{q}(\beta(0))$ satisfies
$\mathbf{q}_{\beta}=\mathbf{q} \circ \beta$.
\end{prop}
\begin{proof}
The time derivatives of the functions $\mathbf{q}\circ \beta$ and $\mathbf{q}_\beta$ coincide, both being equal to $\langle \mathbf{F}(\q_\beta,\mathbf{s}_\beta),\mathbf{s}_\beta'\rangle\beta' $.
Since we also have $\q\circ\beta(0)=\q_\beta(0)$,
the conclusion follows from a direct application of the Cauchy-Lipschitz theorem.
\end{proof}
On Fig.~\ref{pppurcell} are presented some geometric interpretations of what a linear swimmer is. From Prop.~\ref{THE_reparam} above, one can easily deduce:
\begin{figure}
\centerline{\input{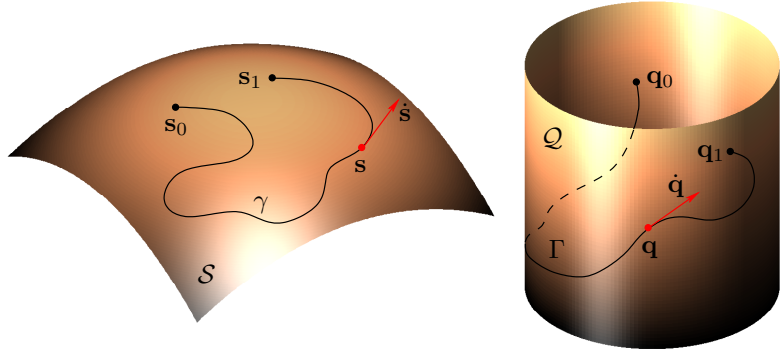}}
\caption{\label{pppurcell}To any path $\gamma$ in the configuration space $\mathcal S$ corresponds a path $\Gamma$ in the space $\mathcal Q$, once the initial point $\mathbf q_0$ has been fixed. Because it is independent of the parameterization of $\gamma$, we can define the mapping $\mathbf s\in\gamma\mapsto\mathbf q(\mathbf s)\in\Gamma$. Moreover, for any time-parameterization, the map $\dot{\mathbf s}\mapsto\dot{\mathbf q}$ is linear.}
\end{figure}
\begin{prop}\label{PRO_flapping}
For any $T>0$, there exists a real number $R>0$ such that
for any  $\mathcal C^1$ function $\beta:\mathbf{R}_+ \rightarrow [0,T]$ and for
any initial condition $\mathbf{q}_0\in\mathcal Q$,  the solution
$\mathbf{q}_{\beta}:\mathbf{R}_+\rightarrow \mathcal Q$ to Equation
(\ref{EQ_main}) corresponding to the shape-changes $t\in\mathbf R_+\mapsto
\mathbf s(\beta(t))\in\mathcal Q$ with initial condition $\mathbf{q}_0$, remains
in the ball of $\mathcal Q$ of center $\mathbf{q}_0$ and radius $R$.
\end{prop}

\begin{proof}
Fix $\mathbf{q}_0\in\mathcal Q$ and denote by $\q$ the solution to Equation
(\ref{EQ_main}) with Cauchy data $\q_0$. The interval $[0,T]$ is compact and hence the set
$\{\q_\beta(\mathbf{R}_+)\,:\, \beta \in \mathcal C^1(\mathbf{R}_+,[0,T])\}=
\{ \q(\beta(\mathbf{R}_+))\,:\, \beta \in \mathcal C^1(\mathbf{R}_+,[0,T])\} = \q([0,T])
$ is also compact (because $\q$ is continuous) and hence bounded by some constant, which in addition can be chosen independently of $\q_0$. Indeed, we have, for any $t\in[0,T]$:
$${\rm d}_{\Q}(\mathbf{q}(t),\mathbf{q}_0 )\leq \int_0^t \left \|\langle \mathbf{F}(\q(u),\s(u)),\dot{\mathbf{s}}(u)\rangle \right \|_{T_{\q(u)}\Q}{\rm d}u
\leq  K \int_0^t \|\mathbf{s}(u)\|_{\mathcal S}\|\dot{\mathbf{s}}(u)\|_{\mathcal S}{\rm d}u.
$$
 Since $t\in\mathbf R_+\mapsto \s(t)\in\mathcal S$ is smooth, the last integral is bounded for every $t$ in $[0,T]$ by $\int_0^T \|\mathbf{s}(u)\|_{\mathcal S}\|\dot{\mathbf{s}}(u)\|_{\mathcal S}{\rm d}u<+\infty$.
\end{proof}

Our topological version of Purcell's scallop theorem will be obtained by reinterpreting Proposition~\ref{PRO_flapping}, 
in the frame of differential geometry, using the classical notion of universal cover (see for instance \citep{Morita:2001aa} for an introduction to covering manifolds), which we now recall the definition: 
For any (finite dimensional) smooth connected Riemannian manifold $X$,
the universal cover of $X$ is a simply connected smooth Riemannian manifold $\widehat{X}$
endowed with a canonical projection $p:\widehat{X}\rightarrow X$ enjoying the following property: 
For every $x$ in $X$ and $y$ in $\widehat{X}$
satisfying $p(y)=x$, there exists a neighborhood $U_y$ of $y$ in $\widehat{X}$ and a neighborhood
$U_x$ of $x$ in $X$ such that $p_{|U_y}:U_y\rightarrow U_x$ be an isometric diffeomorphism. Any vector
field $v$ on $X$ can be lifted to $\widehat{X}$ by defining locally
$\hat{v}(y)=(T_y p)^{-1} v(p(y))$
for any $y$ in $U_y$.
Any curve $\tau:[0,T]\rightarrow X$ solution to the ODE $\dot{\tau}=v(\tau)$ can hence
be lifted to $\widehat{X}$ as well by choosing any base point $y_0$ in
$p^{-1}(\tau(0))$, and by considering the solution to the ODE $\dot{y}=\hat{v}(y)$
with initial condition $y(0)=y_0$.

The Banach structure of $\Shape$ induces a Riemannian structure $g_X$ on $X$.
This Riemannian structure is compatible with the topology of $X$.
Seen as a one dimensional manifold (endowed with its own topology), $X$ can be either compact (or equivalently bounded for $g_X$, and hence diffeomorphic to $\mathcal S^1$, the one dimensional torus) or not (and hence diffeomorphic to $\mathbf{R}$). In both cases, the universal
cover of this manifold is $\mathbf{R}$. With this material, we can restate Proposition \ref{PRO_flapping}
as follows:
\begin{theorem}[Generalized scallop theorem]\label{PRO_topol_flapping}
Consider any smooth shape function  $t\in\mathbf R_+\mapsto\mathbf{s}(t)\in X$ and any lift $\hat{\mathbf{s}}:\mathbf R_+\rightarrow \widehat{X}$ of $\mathbf{s}$ (this choice is unique up to the choice of the base point $\hat{\mathbf{s}}(0)$ in $p^{-1}(\mathbf{s}(0))$).
If the subset $\hat{\mathbf{s}}(\mathbf{R}_+)$ of $\widehat{X}$ is of finite length, or equivalently if the topological closure in $\widehat{X}$ of ${\hat{\mathbf{s}}(\mathbf{R}_+)}$ is compact,
 then any solution $\mathbf{q}:{\mathbf R}_+\rightarrow \Q$ to Equation (\ref{EQ_main}) is bounded as well.
\end{theorem}
On Fig.~\ref{comment_the} and according to the theorem, only the shape-changes relating to the third case can result in locomotion. 
\begin{proof}
Assume that the path $\hat{\mathbf{s}}(\mathbf{R}_+)\subset\widehat X$ is of finite length $l$ (with $l>0$) and denote by $\hat{\tau}:t\in[0,l]\mapsto\hat{\tau}(t)\in\hat{\mathbf{s}}(\mathbf{R}_+)$ its arc-length parameterization.
Then, there exists a smooth function $\beta:\mathbf{R}_+\rightarrow [0,l]$ such that $\mathbf{s}=p\circ \hat{\tau}\circ \beta$ and
the conclusion follows from Proposition \ref{PRO_flapping}.
\end{proof}
The following comments are worth being considered:
\begin{itemize}
\item The geometric hypothesis of the theorem is independent of
 the choice of the base point $\hat{\mathbf{s}}_0$. 
 \item As already mentioned, the case where the shape function
$t\in\mathbf R_+\mapsto{\mathbf{s}}(t)\in\mathcal S$ is not periodic and $\hat{\mathbf s}(\mathbf R_+)$ of infinite length agrees with the hypothesis, whereas it is not covered by Purcell's original theorem.
\item The topological nature of Purcell's scallop theorem has been known for quite a long time. 
For instance, in \citep{Raz:2007aa}, an interpretation of periodic shape-changes is given in term of retract. 
This result could be extended to non periodic and possibly non compact shape-changes by saying that the closure of $\hat{\mathbf s}(\mathbf R_+)$ in $\widehat X$ has to be homotopic to a compact set.
\item In \citep{DeSimone:2009aa}, the theorem is connected to the exactness of some closed differential 1-form. 
Notice that in the simply connected universal cover, exactness and closedness of differential 1-form are actually equivalent. 
\end{itemize}
 \begin{figure}\label{FIG_universal_cover}
 \centerline{\input{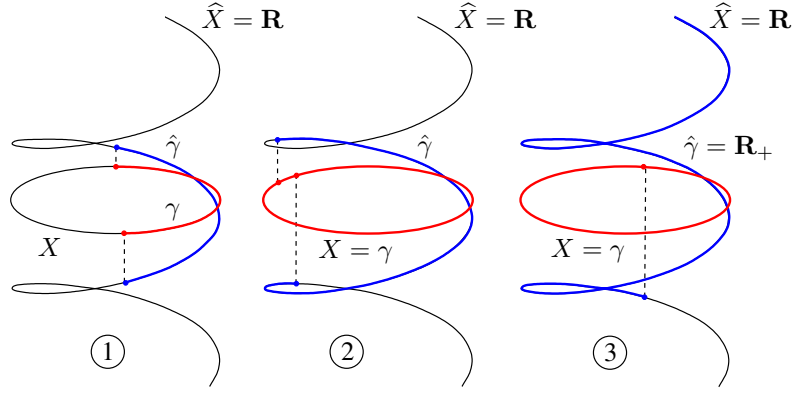}}
 \caption{\label{comment_the}Denote $\hat\gamma=\hat{\mathbf s}(\mathbf R_+)$. In cases 1 and 2, the shape-changes cannot result in locomotion.
Locomotion is possible in the third case only.}
 \end{figure}
\section{Swimmer at low and high Reynolds numbers}\label{SEC_examples}
\label{SEC:modelling}
In this Section we derive the Euler-Lagrange equations for low and high Reynolds
swimmers. We show that, although the properties of the fluid are completely
different in both cases, the equations eventually agree with the general form (\ref{EQ_main}) of linear swimmers.
In the modeling, we will assume that:
\begin{enumerate}[(i)]
\item\label{first_sym} The swimmer is alone in the fluid and the fluid-swimmer system fills the whole space. It entails that all of the positions in the fluid are equivalent and the equations of motion can be written with respect to a frame attached to the swimmer.
\item \label{second_sym}The buoyant force is neglected.
\item \label{third_sym}The fluid-swimmer system is at rest at the initial time.
 \end{enumerate}
\subsection{Kinematics}
The shape-changing body occupies a domain $\mathcal B$ of $\mathbf R^3$ and $\mathcal
F:=\mathbf R^3\setminus \bar{\mathcal B}$ is the domain occupied by the
surrounding fluid.
We consider a Galilean fixed frame $(\mathbf e_1,\mathbf e_2,\mathbf e_3)$ and a
moving frame $(\mathbf e_1^\ast,\mathbf e_2^\ast,\mathbf e_3^\ast)$ attached to
the body.  At any time there
exists $R\in{\rm SO}(3)$ such that $\mathbf e_j^\ast=R\mathbf e_j$ and we assume
that the origin of the latter frame coincides with the center of mass $\mathbf
r\in\mathbf R^3$ of the body. We introduce the notation $\mathbf q:=(R,\mathbf
r)$, which belongs to the Euclidean group ${\mathcal Q}:={\rm
SO}(3)\times\mathbf R^3$. The Eulerian rigid velocity field of the frame
$(\mathbf e_j^\ast)$ with respect to $(\mathbf e_j)$ is defined at any point
$x\in\mathbf R^3$ by $\mathbf w_r(x):=\boldsymbol\omega\times(x-\mathbf
r)+\mathbf v$, where $\mathbf v:=\dot{\mathbf r}$ and $\boldsymbol\omega$ is the
rotation vector defined by $\dot RR^Tx=\boldsymbol\omega\times x$ for all
$x\in\mathbf R^3$.
The shape changes are described by means of a set of diffeomorphisms
$\chi_{\mathbf s}$, indexed by the shape variable $\mathbf{s}$, and that map a
reference domain (let say for instance the unit ball $B$) onto the domain
$\mathcal B^\ast$ of the body as seen by an observer attached to the moving
frame $(\mathbf e_j^\ast)$. The Eulerian velocity at any point $x$ of the
swimmer is the sum of the rigid velocity and the velocity of deformation:
$\mathbf w=\mathbf w_r+\mathbf w_d$ where $\mathbf
w_d:=R\langle\partial_{\mathbf s}\chi_{\mathbf s}(\chi_{\mathbf s}(R^T(x-\mathbf
r))^{-1}),\dot{\mathbf s}\rangle$. It can be expressed in the moving frame:
$\mathbf w^\ast=\mathbf w_r^\ast+\mathbf w_d^\ast$ where $\mathbf
w_r^\ast:=\boldsymbol\omega^\ast\times (x^\ast)^T+\mathbf v^\ast$, $\mathbf
v^\ast=R^T\mathbf v$, $x^\ast=R^T(x-\mathbf r)$,
$\boldsymbol\omega^\ast:=R^T\boldsymbol\omega$ and $\mathbf
w_d^\ast=\langle\partial_{\mathbf s}\chi_{\mathbf s}(x^\ast),\dot{\mathbf
s}\rangle$  (more generally, quantities will be denoted with an asterisk when
expressed in the moving frame). The deformation tensor is $\mathbb F_{\mathbf
s}:=D\chi_{\mathbf s}$ and, keeping the classical notation of Continuum
Mechanics, we introduce $J_{\mathbf s}:=|\det(\mathbb F_{\mathbf s})|$.
 \begin{figure} \label{FIG_Kinematic}
 \centerline{\input{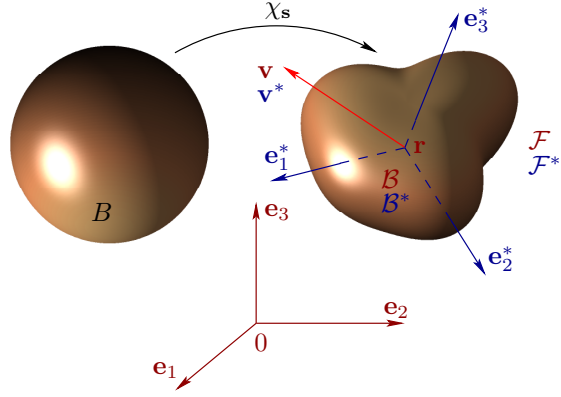}}
 \caption{Kinematics of the model: The Galilean frame $(\mathbf e_j)$ and the
 moving frame $(\mathbf e_j^\ast)$ with $\mathbf e_j^\ast=R\mathbf e_j$
 ($R\in{\rm SO}(3)$). Quantities are denoted with an asterisk when expressed in
 the moving frame. The domain of the body $\mathcal B^\ast$ is the image of the
 unit ball $B$ by a diffeomorphism $\chi_{\mathbf s}$ ($\mathbf s\in\mathcal S$
 is the shape variable) and $\mathcal F$ is the domain of the fluid. The center
 of mass of the body is denoted $\mathbf r$ and $\mathbf v$ is its velocity.}
 \end{figure}
\subsection{Dynamics}
The density of the body can be deduced from a given constant density
$\varrho_0>0$, defined in $B$, according to the conservation of mass principle:
$\varrho^\ast\circ\chi_{\mathbf s}=\varrho_0/J_{\mathbf s}$. The volume of the
swimmer is ${\rm Vol}=\int_BJ_{\mathbf s}(x){\rm d}x$, its mass $m=\varrho_0{\rm
Vol}$ and its inertia tensor $\mathbb I_{\mathbf s}^\ast=\varrho_0\int_B|\chi_{\mathbf
s}|^2\mathbb I{\rm d}-\chi_{\mathbf s}\otimes\chi_{\mathbf s}{\rm d}x$ in
$(\mathbf e_j^\ast)$ and $\mathbb I_{\mathbf s}=R\mathbb I^\ast_{\mathbf s}R^T$
in $(\mathbf e_j)$.  The deformations have to result from the work of internal forces
within the body. It means that in the absence of fluid, the swimmer is not
able to modify its linear and angular momenta. Assuming that the swimmer is at
rest at some instant, we deduce that at any time
$\int_{B}\langle\partial_{\mathbf s}\chi_{\mathbf s},\dot{\mathbf s}\rangle{\rm
d}x=0$ and $\int_{B}\langle\partial_{\mathbf s}\chi_{\mathbf s},\dot{\mathbf
s}\rangle\times\chi_{\mathbf s}{\rm d}x=0$. These equations have to be
understood as constraints on the shape variable and will be termed subsequently
the {\it self-propulsion hypotheses}. The fluid obeys, in the whole generality,
to the Navier-Stokes equations for incompressible fluid:
$\varrho_f\frac{D}{Dt}\mathbf u-\nabla\cdot{\mathbb T}=0$ and
$\nabla\cdot\mathbf u=0$ in $\mathcal F$ for all $t>0$ ($\varrho_f$ is the
fluid's density,  $\mathbf u$ the Eulerian velocity, $D/Dt$ the convective
derivative, ${\mathbb T} := 2\mu D(\mathbf u)-p\mathbb I{\rm d}$ with $D(\mathbf
u):=(1/2)(\nabla\mathbf u+\nabla\mathbf u^T)$ is the stress tensor and $\mu$ the
dynamic viscosity). The rigid displacement of the body is governed by
Newton's laws for linear and angular momenta: $m\frac{d}{dt}{\mathbf
v}=-\int_{\partial{\mathcal B}}{\mathbb T} n\,{\rm d}\sigma$ and
$\frac{d}{dt}(\mathbb I_{\mathbf s}\boldsymbol\omega)=-\int_{\partial\mathcal
B}{\mathbb T} n \times(x-\mathbf r)\,{\rm d}\sigma$ (the rigid displacement is
caused  by the hydrodynamical forces only) where $n$ is the unit vector to $\partial\mathcal B$ directed towards the interior of $\mathcal B$. These equations have to be supplemented
with boundary conditions on $\partial\mathcal B$, which can be either $\mathbf
u\cdot n =\mathbf w\cdot n$ (slip boundary conditions) or $\mathbf u =\mathbf w$
(no-slip boundary conditions) and with initial data: $\mathbf u(0)=\mathbf u_0$,
$R(0)=R_0$, $\mathbf r(0)=\mathbf r_0$,
$\boldsymbol\omega(0)=\boldsymbol\omega_0$ and $\mathbf v(0)=\mathbf v_0$.

We focus on two limit problems connecting to the value of the Reynolds number
${\rm Re}:=\varrho VL/\mu$ ($V$ is the mean fluid velocity and $L$ is a
characteristic linear dimension). The first case ${\rm Re}\ll1$ concerns low Reynolds swimmers like bacteria (or more generally so-called micro swimmers whose size is
about $1\mu m$). For the second ${\rm Re}\gg1$, we will restrain our study to
irrotational flows (i.e. $\nabla\times\mathbf u=0$) and so it is relevant for
large  animals swimming quite slowly, a case where vorticity can be neglected.
\subsection{Low Reynolds swimmers}
For micro-swimmers, scientists agree that inertia (for both the fluid and the
body) can be neglected in the dynamics. It means that in the modeling, we can
set $\varrho_0=\varrho_f=0$. In this case, the Navier-Stokes equations reduce to
the steady Stokes equations $-\nabla\cdot{\mathbb T} =0$, $\nabla\cdot\mathbf u=0$ and
we choose no-slip boundary conditions $\mathbf u=\mathbf w$ on $\partial\mathcal
B$. Introducing $\mathbf u^\ast(x^\ast):=R^T\mathbf u(Rx^\ast+\mathbf r)$ and
$p^\ast(x^\ast)=p(Rx+\mathbf r)$, the equations keep the same form when expressed in the frame $(\mathbf e_1^\ast,\mathbf e_2^\ast)$,
namely: $-\nabla\cdot{\mathbb T}^\ast=0$, $\nabla\cdot\mathbf u^\ast=0$ in
$\mathcal F^\ast$ with boundary data: $\mathbf u^\ast=\mathbf w^\ast$. From a
mathematical point of view, the main advantage is that the equations are now
linear. Notice that since the equations are stationary, no initial data is
required for the fluid. Newton's laws read $\int_{\partial\mathcal
B^\ast}{\mathbb T}^\ast n{\rm d}\sigma=0$ and $\int_{\partial\mathcal
B^\ast}{\mathbb T}^\ast n\times x^\ast{\rm d}\sigma=0$ (it means that
the system fluid-swimmer is in equilibrium at every moment. Indeed, since
there is no mass, any force would produce an infinite acceleration). As
already mentioned, the Stokes equations are linear. It entails that the solution
$(\mathbf u^\ast,p^\ast)$ is linear with respect to the boundary data $\mathbf
w^\ast$ and we draw the same conclusion for the stress tensor $\sigma^\ast$
because it is linear in $(\mathbf u^\ast,p^\ast)$. Observe now that $\mathbf
w^\ast$ is linear in the 3 components $\omega^\ast_j$ ($j=1,2,3$) of
$\boldsymbol\omega^\ast$, in the 3 components $v_j^\ast$ ($j=1,2,3$) of $\mathbf
v^\ast$ and in $\dot{\mathbf s}$. We can then decompose any solution to the
Stokes equations accordingly: $\mathbf u^\ast=\sum_{j=1}^3\omega^\ast_j\mathbf
u^\ast_j+v^\ast_j\mathbf u^\ast_{j+3}+\langle\mathbf u_d^\ast,\dot{\mathbf
s}\rangle$, $p^\ast=\sum_{j=1}^3\omega^\ast_j
p^\ast_j+v^\ast_jp^\ast_{j+3}+\langle p_d^\ast,\dot{\mathbf s}\rangle$ and the
stress tensor as well: $\mathbb T^\ast=\sum_{j=1}^3\omega^\ast_j\mathbb
T^\ast_j+v^\ast_j\mathbb T^\ast_{j+3}+\langle\mathbb T_d^\ast,\dot{\mathbf
s}\rangle$. Notice that the {\it elementary solutions} $(\mathbf
u_j^\ast,p_j^\ast)$ as well as the {\it elementary stress tensors} $\mathbb
T^\ast_j$ depend on the shape variable $\mathbf s$ only. We next introduce the
$6\times 6$ matrix $\mathbb M^r(\mathbf s)$ whose entries $M^r_{ij}(\mathbf s)$
are $M_{ij}^r(\mathbf s):=\int_{\partial\mathcal B^\ast}\mathbf
e^\ast_i\cdot(\mathbb T_j^\ast n\times x^\ast){\rm
d}\sigma=\int_{\partial\mathcal B^\ast}(x^\ast\times \mathbf
e^\ast_i)\cdot\mathbb T^\ast_jn{\rm d}\sigma$ ($1\leq i\leq 3$, $1\leq j\leq 6$)
and
$M_{ij}^r(\mathbf s):=\int_{\partial\mathcal B^\ast}\mathbf
e^\ast_{i-3}\cdot\mathbb T_j^\ast n{\rm d}\sigma$ ($4\leq i\leq 6$, $1\leq j\leq
6$) and $\mathbb N(\mathbf s)$, the linear continuous map from $\mathcal S$ into
$\mathbf R^6$ defined by $\langle \mathbb N({\mathbf s}),\dot{\mathbf
s}\rangle:=(\int_{\partial\mathcal B^\ast}\langle\mathbb T_d^\ast,\dot{\mathbf
s}\rangle n\times x^\ast{\rm d}\sigma,\int_{\partial\mathcal
B^\ast}\langle\mathbb T_d^\ast,\dot{\mathbf s}\rangle n{\rm d}\sigma)$.  We can
rewrite Newton's laws as $\mathbb M^r({\mathbf s})\dot{\mathbf
q}^\ast+\langle\mathbb N({\mathbf s}),\dot{\mathbf s}\rangle=0$ where
$\dot{\mathbf q}^\ast:=(\boldsymbol\omega^\ast,\mathbf v^\ast)^T\in\mathbf R^6$.
Upon an integration by parts, we get the equivalent definition for he entries of
the matrix $\mathbb M^r(\mathbf s)$: $M_{ij}^r(\mathbf s):=2\mu\int_{\mathcal
B^\ast}D(\mathbf u^\ast_i):D(\mathbf u^\ast_j){\rm d}x^\ast$, whence we deduce
that $\mathbb M^r(\mathbf s)$ is symmetric and positive definite. The same
arguments for $\mathbb N(\mathbf s)$ lead to the identity: $(\langle \mathbb
N(\mathbf s),\dot{\mathbf s}\rangle)_j=2\mu\int_{\mathcal B^\ast}D(\mathbf
u^\ast_j):D(\langle \mathbf u^\ast_d(\mathbf s),\dot{\mathbf s}\rangle){\rm
d}x^\ast$. We eventually obtain the Euler-Lagrange equation: $\dot{\mathbf
q}^\ast=-\mathbb M^r(\mathbf s)^{-1}\langle \mathbb N(\mathbf s),\dot{\mathbf
s}\rangle$, or equivalently $\dot{\mathbf q}=-\mathbb R(\mathbf q)\mathbb
M^r(\mathbf s)^{-1}\langle \mathbb N(\mathbf s),\dot{\mathbf s}\rangle$ where
$\mathbb R(\mathbf q):={\rm diag}(R,R)$ and $\dot{\mathbf
q}:=(\boldsymbol\omega,\mathbf v)^T$. Although this modeling is not new, the
authors were not able to find the Euler-Lagrange equation in this particular
form, allowing one in particular to deduce the following result:
\begin{prop}
The dynamics of a micro-swimmer is independent of  the viscosity of the fluid.
Or, in other words, the same shape changes produce the same rigid displacement,
whatever the viscosity of the fluid is.
\end{prop}
\begin{proof}
Let $(\mathbf u_j^\ast,p^\ast_j)$ be an elementary solution (as defined in the
modeling above) to the Stokes equations corresponding to a viscosity $\mu>0$,
then $(\mathbf u_j^\ast,(\tilde\mu/\mu) p^\ast_j)$ is the same elementary
solution corresponding to an other viscosity $\tilde\mu>0$. Since the
Euler-Lagrange equation depends only on the Eulerian velocities $\mathbf
u_j^\ast$, the proof is completed.
\end{proof}

\subsection{High Reynolds swimmers}
Assume now that the inertia is preponderant with respect to the viscous force
(it is the case when ${\rm Re}\ll 1$).
The Navier-Stokes equations simplify into the Euler equations:
$\varrho_f\frac{D}{Dt}\mathbf u-\nabla\cdot\mathbb T=0$, $\nabla\cdot\mathbf
u=0$ in $\mathcal F$ where $\mathbb T=-p{\rm I}d$ and we specify the  boundary
conditions to be: $\mathbf u\cdot n=\mathbf w\cdot n$  on $\partial\mathcal B$
(slip boundary conditions). Like in the preceding Subsection, we will assume that
at some instant, the fluid-body system is at rest. According to Kelvin's
circulation theorem, if the flow is irrotational at some moment (i.e.
$\nabla\times\mathbf u=0$) then, it has always been (and will always remain)
irrotational. We can hence suppose that $\nabla\times\mathbf u=0$ for all times
and then, according to the Helmholtz decomposition, that there exists for all
time $t>0$ a potential scalar function $\varphi$ defined in $\mathcal F$, such
that $\mathbf u=\nabla\varphi$. The divergence-free condition leads to
$\Delta\varphi=0$ and the boundary condition reads: $\partial_n\varphi=\mathbf
w\cdot n$. Following our rule of notation, we introduce the function
$\varphi^\ast(t,x^\ast):=\varphi(t,R^T(x-\mathbf r))$ ($t>0$, $x^\ast\in\mathcal
F^\ast$), which is harmonic and satisfies $\partial_n\varphi^\ast=\mathbf
w^\ast\cdot n$ on $\partial\mathcal B^\ast$. The potential $\varphi^\ast$ is
linear in $\mathbf w^\ast$, so it can be decomposed into
$\varphi^\ast=\sum_{j=1}^3\omega^\ast_j\varphi^\ast_j+v_j^\ast\varphi^\ast_{j+3}
+\langle\varphi^\ast_d,\dot{\mathbf s}\rangle$ (this process is usually referred
to as Kirchhoff's law). At this point, we do not invoke Newton's laws to derive
the Euler-Lagrange equation but rather use the formalism of Analytic Mechanics.
Both approaches (Newton's laws of Classical Mechanics and the Least Action
principle of Analytic Mechanics) are equivalent (as proved in \citep{Munnier:2008ab}),
but the latter is notably simpler and shorter. In the absence of buoyant force, the
Lagrangian function $\mathcal L$ of the body-fluid system coincides with the
kinetic energy: $\mathcal L=m\frac{1}{2}|\mathbf
v^\ast|^2+\frac{1}{2}\boldsymbol\omega^\ast\cdot\mathbb I_{\mathbf
s}^\ast\boldsymbol\omega^\ast+\frac{1}{2}\int_{\mathcal
B^\ast}\varrho^\ast|\mathbf w_d^\ast|^2{\rm d}x^\ast+\frac{1}{2}\int_{\mathcal
F^\ast}\varrho_f|\mathbf u^\ast|^2{\rm d}x^\ast$. In this sum, one can identify,
from the left to the right: the kinetic energy of the body connecting to the
rigid motion (two first terms), the kinetic energy resulting from the
deformations and the kinetic energy of the fluid. We can next compute that:
$\int_{\mathcal B^\ast}\varrho^\ast|\mathbf w_d^\ast|^2{\rm
d}x^\ast=\int_{B}\varrho_0|\langle\partial_{\mathbf s}\chi_{\mathbf
s},\dot{\mathbf s}\rangle|^2{\rm d}x$ (upon a change of variables) and
$\int_{\mathcal F^\ast}\varrho_f|\mathbf u^\ast|^2{\rm d}x^\ast=\int_{\mathcal
F^\ast}\varrho_f|\nabla\varphi^\ast|^2{\rm d}x^\ast$. It leads us to introduce
the so-called mass matrices $\mathbb M^r_f(\mathbf s)$, whose entries
$(M^r_f)_{ij}(\mathbf s)$ are defined by $(M^r_f)_{ij}(\mathbf
s):=\int_{\mathcal
F^\ast}\varrho_f\nabla\varphi^\ast_i\cdot\nabla\varphi^\ast_j{\rm d}x^\ast$
($1\leq i,j\leq 6$), and $\mathbb M^r(\mathbf s):={\rm diag}(\mathbb
I^\ast_{\mathbf s},m\mathbb I{\rm d})+\mathbb M^r_f(\mathbf s)$. One easily
checks that $\mathbb M^r(\mathbf s)$ is symmetric and positive definite. We
define as well the linear map $\mathbb N(\mathbf s)$ from $\mathcal S$ into
$\mathbf R^6$ by $(\langle\mathbb N(\mathbf s),\dot{\mathbf
s}\rangle)_j:=\int_{\mathcal
F^\ast}\varrho_f\nabla\varphi^\ast_j\cdot\nabla\langle\varphi^\ast_d,\dot{
\mathbf s}\rangle{\rm d}x^\ast$ ($1\leq j\leq 6$) and we can rewrite the kinetic
energy of the fluid in the form: $\frac{1}{2}\dot{\mathbf q}^\ast\cdot\mathbb
M^r_f(\mathbf s)\dot{\mathbf q}^\ast+\dot{\mathbf q}^\ast\cdot\langle\mathbb
N(\mathbf s),\dot{\mathbf s}\rangle$. Invoking now the Least Action principle,
we claim that the Euler-Lagrange equation is: $\delta_L\mathcal L=0$ where we
have denoted $\delta_L:=\frac{d}{dt}\frac{\partial}{\partial\dot{\mathbf
q}}-\frac{\partial}{\partial\mathbf q}$ the Lagrangian differential operator
connecting to the system of generalized coordinates $(\mathbf q,\dot{\mathbf
q})$. Introducing the impulses $(\boldsymbol\Pi,\mathbf P)^T:=\mathbb
M^r(\mathbf s)(\boldsymbol\omega^\ast,\mathbf v^\ast)^T$ and
$(\boldsymbol\Lambda,\mathbf L)^T:=\langle\mathbb N(\mathbf s),\dot{\mathbf
s}\rangle$ (homogeneous to momenta) and since $\langle\delta_L\dot{\mathbf
q}^\ast,\dot{\mathbf
Q}\rangle=(\boldsymbol\Omega^\ast\times\boldsymbol\omega^\ast,
\boldsymbol\Omega^\ast\times\mathbf v^\ast-\boldsymbol\omega^\ast\times\mathbf
V^\ast)^T$ for any $\dot{\mathbf Q}:=(\boldsymbol\Omega,\mathbf V)^T\in\mathbf
R^6$ (and $\dot{\mathbf Q}^\ast:=(\boldsymbol\Omega^\ast,\mathbf V^\ast)^T$ with
$\boldsymbol\Omega^\ast:=R^T\boldsymbol\Omega$, $\mathbf V^\ast=R^T\mathbf V$),
we deduce that $\langle\delta_L\mathcal L,\dot{\mathbf
Q}\rangle=\frac{d}{dt}(\boldsymbol\Pi+\boldsymbol\Lambda,\mathbf P+\mathbf
L)\cdot\dot{\mathbf Q}^\ast+(\boldsymbol\Pi+\boldsymbol\Lambda,\mathbf P+\mathbf
L)\cdot(\boldsymbol\Omega^\ast\times\boldsymbol\omega^\ast,
\boldsymbol\Omega^\ast\times\mathbf v^\ast-\boldsymbol\omega^\ast\times \mathbf
V^\ast)$. The Euler-Lagrange equation is hence the system of ODEs:
$\frac{d}{dt}
(\boldsymbol\Pi+\boldsymbol\Lambda)=(\boldsymbol\Pi+\boldsymbol\Lambda)\times
\boldsymbol\omega^\ast+(\mathbf P+\mathbf L)\times\mathbf v^\ast$ and
$\frac{d}{dt}(\mathbf P+\mathbf L)=(\mathbf P+\mathbf
L)\times\boldsymbol\omega^\ast$. Since we have assumed that at some instant, the
fluid-body system is at rest, we deduce that
$(\boldsymbol\Pi+\boldsymbol\Lambda)=0$ and $(\mathbf P+\mathbf L)=0$ for all
time (this solution is an obvious solution to the differential system), which
can eventually be rewritten as: $\dot{\mathbf q}^\ast=-\mathbb M^r(\mathbf
s)^{-1}\langle \mathbb N(\mathbf s),\dot{\mathbf s}\rangle$, or equivalently
$\dot{\mathbf q}=-\mathbb R(\mathbf q)\mathbb M^r(\mathbf s)^{-1}\langle \mathbb
N(\mathbf s),\dot{\mathbf s}\rangle$.


\subsection{Breaking the symmetry}
Both models' responses to symmetry breaking are very different. If we assume that there is a rigid fixed obstacle in the fluid or that the fluid-swimmer system is confined in a bounded domain, the Euler-Lagrange equation for the low Reynolds swimmer still agrees with the form \eqref{EQ_main} and the scallop theorem still holds true. Indeed, although the matrix $\mathbb M^r(\mathbf s)$ is no longer independent of the position $\mathbf q$ and has to be rather denoted $\mathbb M^r(\s, \mathbf q)$, the dynamics still reads $\dot{\mathbf q}=-\mathbb R(\q)\mathbb M^r(\mathbf
s,\q)^{-1}\langle \mathbb N(\mathbf s),\dot{\mathbf s}\rangle$. Things begin to turn bad when additional degrees of freedom enter the game. It is the case when several swimmers are involved (this case is treated in \citep{Lauga:2008aa}), when there is a moving rigid obstacle or when the swimmer is close to flexible walls. 

The high Reynolds swimmer is much more sensible to the relaxation of the hypotheses (\ref{first_sym}-\ref{third_sym}) and actually if any of these assumptions fails to be true, the Euler-Lagrange equation turns into a second order ODE containing a drift term. Obviously, the scallop theorem fails to apply in this case. We refer to \citep{Munnier:2008ab}  and  \citep{Munnier:2010aa} for details.

\section{When Purcell's scallop can swim... in a perfect fluid}
In Section~\ref{SEC:modelling}, we have derived the Euler-Lagrange equations for
both a low and high Reynolds swimmers. In both models, we have assumed that the
fluid-body system was filling the whole space (the only boundary of the fluid
was the one shared with the swimmer). In a potential flow, this hypothesis is necessary for the
Euler-Lagrange equation to have the particular form required in the statement
of  Theorem~\ref{PRO_flapping}. In this Section, we aim to show, through a numerical example, than the scallop theorem no longer
holds true when the fluid contains in addition to the swimmer, a fixed obstacle.
So, we consider the simple example of the scallop (as modeled in the original
article of Purcell) swimming in a perfect fluid with irrotational flow. Simulations have been realized with the Biohydrodynamics Matlab Toolbox 
(which is free, distributed under license GPL and can be downloaded at \url{http://bht.gforge.inria.fr/}).
The scallop is made of two rigid ellipses linked together by a hinge. As shown
on Fig~\ref{FIG_nage_scallop}, the animal is located close by a rectangular obstacle.
This obstacle breaks the symmetry of the model as described in
Section~\ref{SEC:modelling} and although the scallop is only able to flap, each
stroke will make it get closer to the obstacle, until it eventually
collides with it (see Fig.~\ref{figure3}). A more physical explanation is that during a strokes, the
fluid pressure (on which solely relies the hydrodynamical forces) is
weaker between the obstacle and the scallop's left arm.
\arraycolsep = 0cm
\begin{figure}
     \centering
     \begin{tabular}{|c|c|c|}
      \hline
{\includegraphics[width=.31\textwidth]{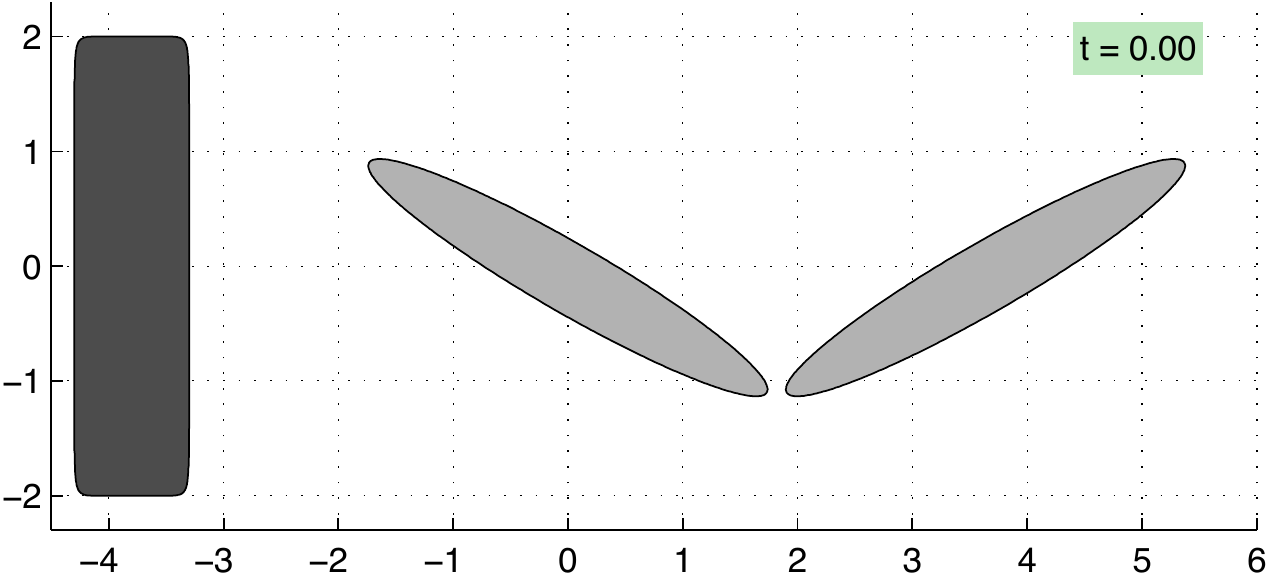}}&
  {\includegraphics[width=.31\textwidth]{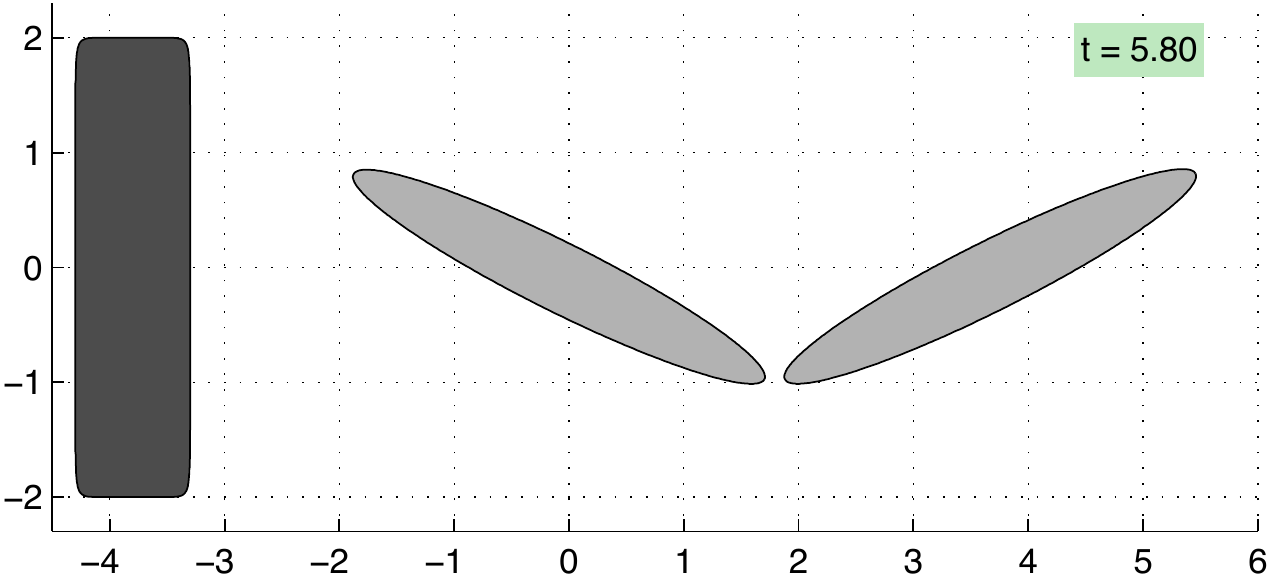}}&
  {\includegraphics[width=.31\textwidth]{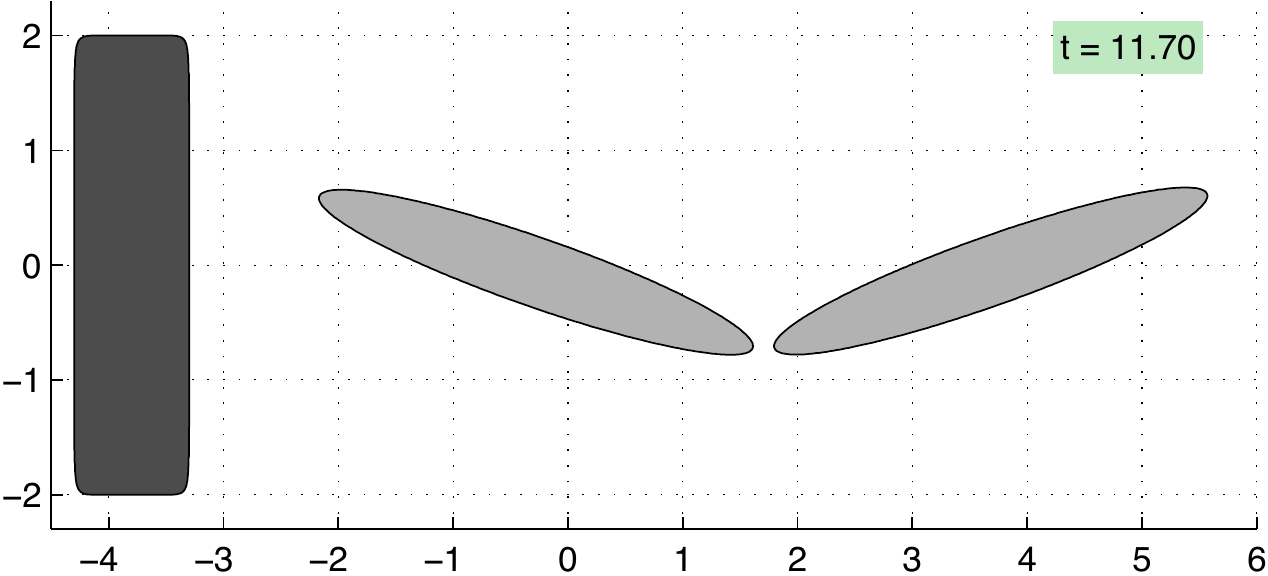}}\\
     \hline
{\includegraphics[width=.31\textwidth]{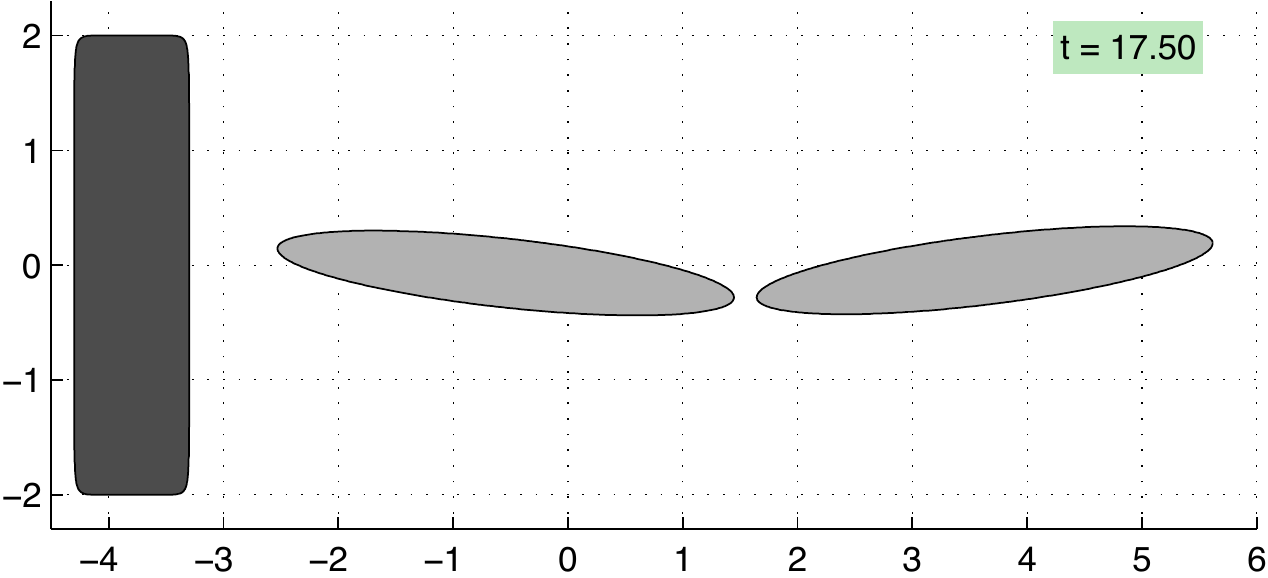}}&
  {\includegraphics[width=.31\textwidth]{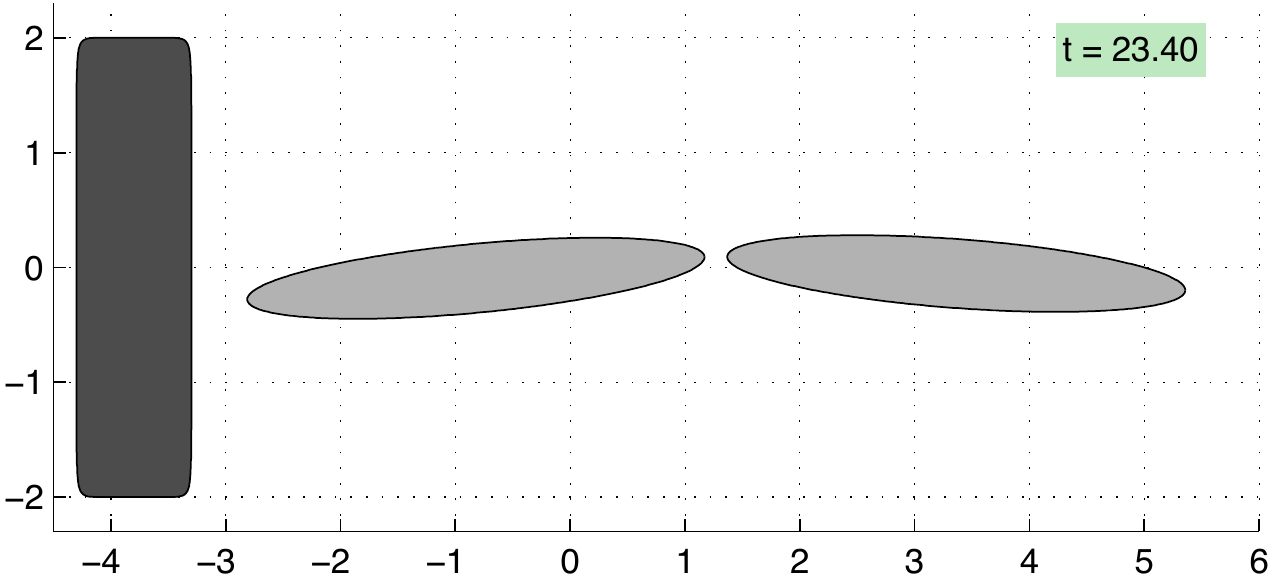}}&
  {\includegraphics[width=.31\textwidth]{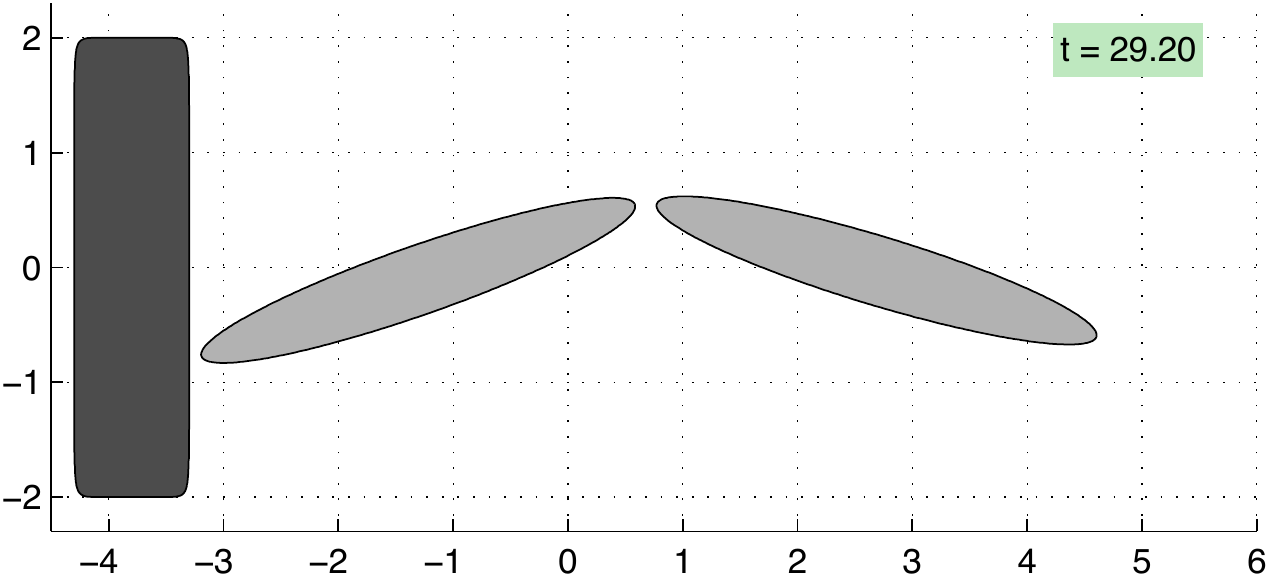}}\\
  \hline
   \end{tabular}
 \caption{\label{FIG_nage_scallop}Purcell's scallop example (in a perfect fluid with
potential flow). The scallop is modeled as an articulated body consisting of two ellipses
linked together by a hinge. The angle between the ellipses is
$\alpha(t)=\pi/3\cos(t)$ ($t>0$). The flapping of the scallop does not
 produce locomotion (after completing a stroke, the scallop
comes back to its exact initial position) in a fluid free of obstacles. However, this is no longer true in
this example where a fixed immersed rigid solid breaks the symmetry of the model.
The flapping motion generates a low pressure zone between the scallop's left arm
and the obstacle, causing the animal to be attracted by the solid and
eventually to collide with it.}
 \end{figure}

 \begin{figure}
     \centering
     \begin{tabular}{|c|c|c|}
     \hline
{\includegraphics[width=.49\textwidth]{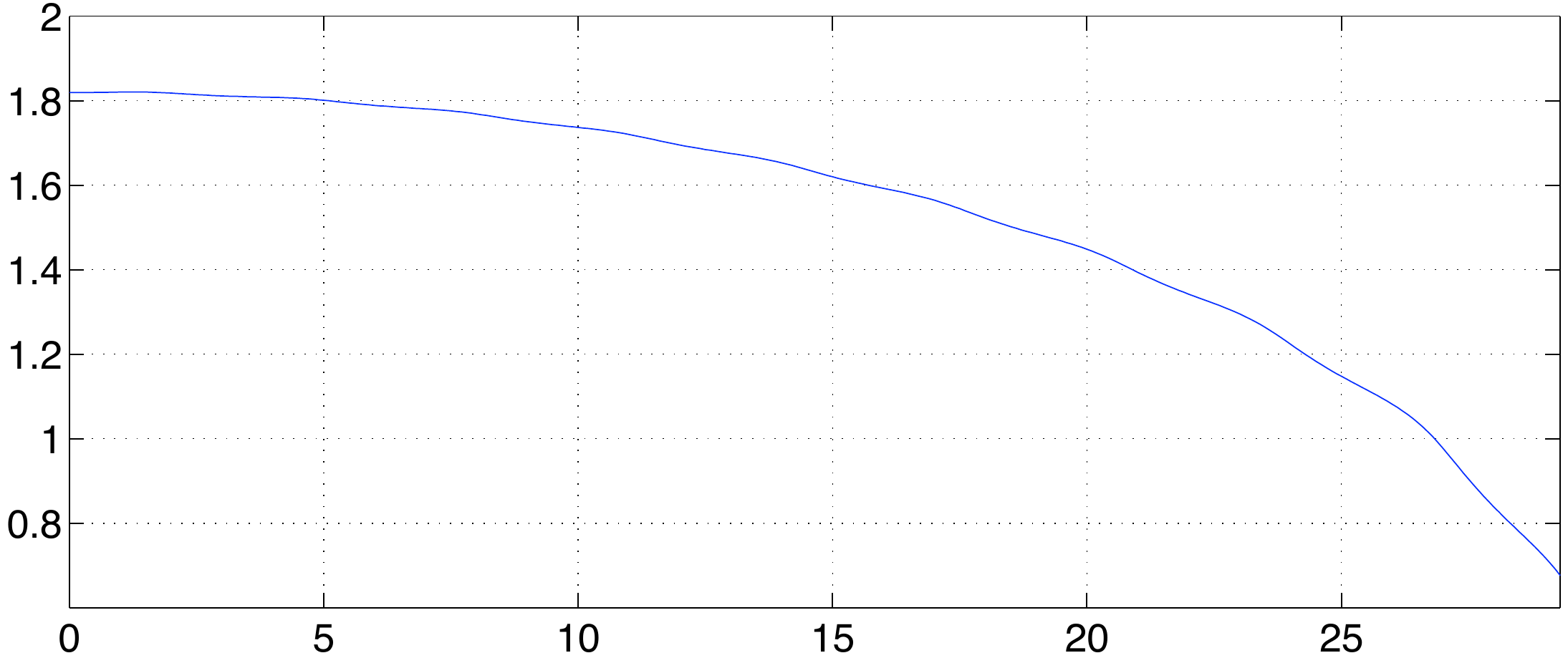}}&
  {\includegraphics[width=.49\textwidth]{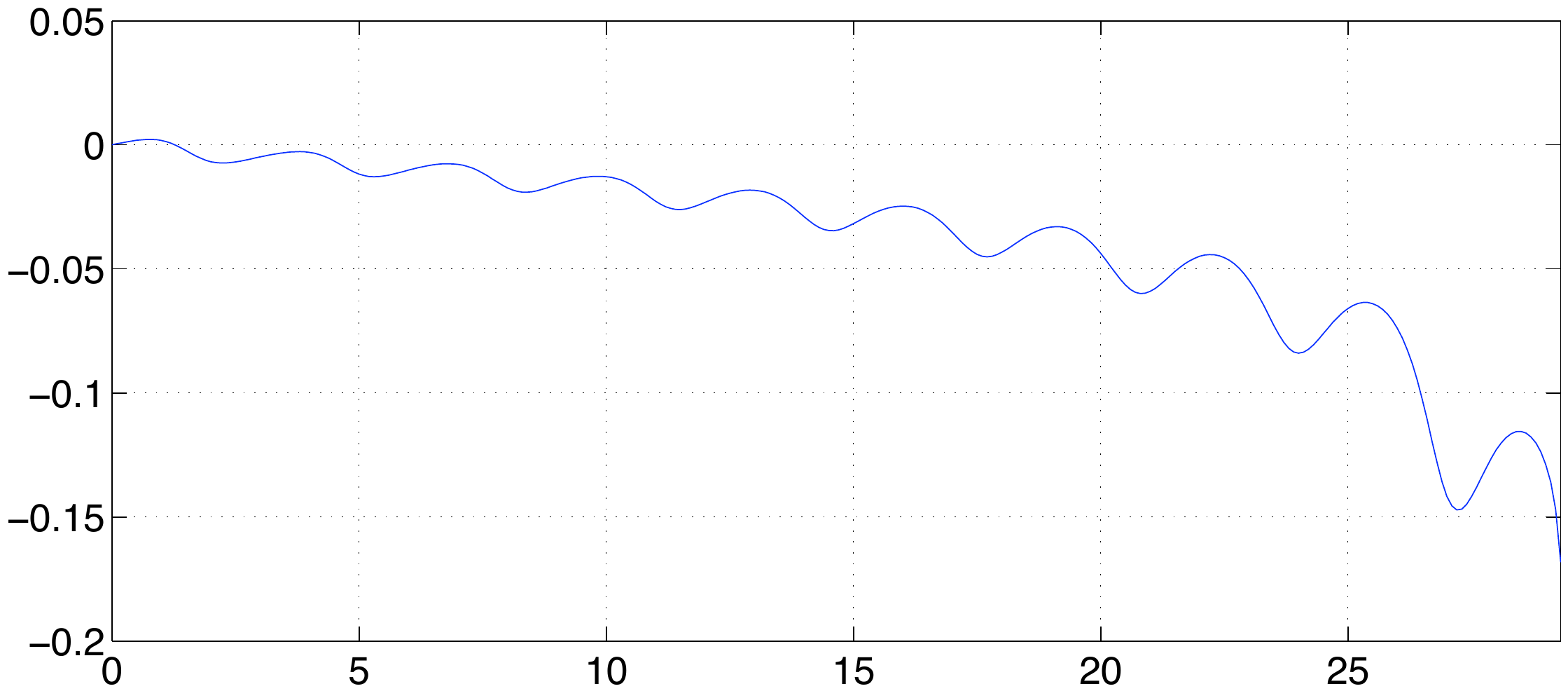}}\\
  \hline
  \end{tabular}
 \caption{\label{figure3}On the left hand side of the figure is plotted the
$x$-coordinate of the center of mass of the scallop with respect to time and on
the right hand side, the $x$-coordinate of its velocity with respect to time.
Owning to the presence of the fixed obstacle, the scallop undergoes a net
displacement to the left. Notice that the velocity increases along with the number of strokes.}
 \end{figure}

\section{Conclusion}
In this article, we have revisited Purcell's scallop theorem and proved that the common hypotheses on the sequence of shape-changes: time periodicity and time reversal invariance, although quite intuitive, are irrelevant from a mathematical point of view and have to be replaced by purely geometric conditions involving the universal cover of the configuration space. We have also shown that Purcell's result applies as well to swimmers at high Reynolds numbers and does not rely solely on the system's inertialess.

\bibliographystyle{abbrvnat}

\nocite{*}
\bibliography{bibli_scallop}

\end{document}